\documentclass[11pt,a4paper]{article}
\usepackage[utf8]{inputenc}
\usepackage{xcolor}
\usepackage{amsthm}
\usepackage{graphicx}
\usepackage{amssymb,amsmath}
\usepackage{comment, hyperref}
\usepackage[ruled, lined, linesnumbered, commentsnumbered, longend]{algorithm2e}
\usepackage{tikz,color}
\usepackage{float}
\usepackage{microtype}
\usepackage{thm-restate}

\usepackage[a4paper, margin=3.5 cm]{geometry}

\newtheorem{definition}{Definition}
\newtheorem{theorem}{Theorem}

\title{On Finality in Blockchains}

\author{Emmanuelle Anceaume$^1$, Antonella Del Pozzo$^2$,\\ Thibault Rieutord$^2$, Sara Tucci-Piergiovanni$^2$}
\date{%
	$^1$CNRS, Univ Rennes, Inria, IRISA, Rennes, France\\%
	$^2$Université Paris-Saclay, CEA, List, F-91120, Palaiseau, France\\[2ex]%
}

\begin{document}

\maketitle
\begin{abstract}
There exist many forms of Blockchain finality conditions, from deterministic to probabilistic terminations. To favor availability against consistency in the face of partitions, most blockchains only offer probabilistic eventual finality: blocks may be revoked after being appended to the blockchain, yet with decreasing probability as they sink deeper into the chain. Other blockchains favor consistency by leveraging the immediate finality of Consensus -- a block appended is never revoked -- at the cost of additional synchronization.

In this paper, we focus on necessary and sufficient conditions to implement a blockchain with deterministic eventual finality, which ensures that selected main chains at different processes share a common increasing prefix. This is a much weaker form of finality that allows us to provide a solution in an asynchronous system subject to unlimited number of byzantine failures. We study stronger forms of eventual finality as well and show that it is unfortunately impossible to provide a \emph{bounded displacement}. By bounded displacement we mean that the (unknown) number of blocks that can be revoked from the current blockchain is bounded. This problem reduces to consensus or eventual consensus depending on whether the bound is known or not. We also show that the classical selection mechanism, such as in Bitcoin, that appends blocks at the longest chain is not compliant with a solution to eventual finality.
\end{abstract}

\section{Introduction}
\label{sec:introduction}
This paper focuses on blockchain finality, which refers to  the time when it becomes impossible to remove a block that has previously been appended to the blockchain. Blockchain finality can be deterministic or probabilistic, immediate or eventual. Informally, immediate finality guarantees,  as its name suggests, that when a block is appended to a local copy, it  is immediately finalized and thus will never be revoked in the future. Most of the permissionned  blockchains satisfy the deterministic form of immediate consistency, including Red Belly blockchain~\cite{redbelly17} and  Hyperledger Fabric blockchain~\cite{hyperledger},  while the probabilistic form of immediate consistency  is typically achieved by permissionless pure proof-of-stake blockchains such as Algorand~\cite{algorand}. Designing  blockchains with immediate finality favors consistency against availability in presence of transient partitions of the system. It leverages the properties of Consensus (i.e a decision value is unique and agreed by everyone), at the cost of synchronization constraints. 
On the other hand eventual finality  ensures that all local copies of the blockchain share a common increasing prefix, and thus finality of  their blocks increases as more blocks are appended to the blockchain. The majority of cryptoassets blockchains, with Bitcoin~\cite{Nakomoto_2008} and Ethereum~\cite{Ethereum} as celebrated examples, guarantee   eventual finality with some probability: blocks may be revoked after being appended to the blockchain, yet with decreasing probability as they sink deeper into the chain. 

Formalization of blockchains in the lens of distributed computing has been recognized as an extremely important topic~\cite{Herlihy2017}. Garay et al.~\cite{GarayKL15} have been the first to analyse the Bitcoin backbone protocol and to define invariants  this protocol has to satisfy to verify with high probability an eventual consistent prefix, i.e. probabilistic eventual finality.    The authors analyzed the protocol in a synchronous system, while Pass et al. \cite{pass2017} extended this line of work considering a more adversarial network. Anta et al. \cite{AF2018} proposed a  formalization of distributed ledgers modeled as  an ordered list of records along with implementations for sequential consistency and linearizability using a total order broadcast abstraction. Not related to the blockchain data structure, authors of \cite{Guerraoui19} formalized the notion of  cryptocurrency showing that Consensus is not needed.  

While probabilistic eventual finality has been widely studied in the context of Bitcoin~\cite{GarayKL15,cachin2017,pass2017}, some few studies have started to lay the foundations of  the computation power of blockchains with deterministic eventual finality consistency.
Anceaume et al. \cite{ADLPT19} have been the first to propose a formal blockchain specification as a composition of abstract data types all together with a hierarchy of consistency criteria. This work captured the  convergence process of two distinct classes of blockchain systems: the class providing strong prefix (for each pair of chains returned at two different processes, one is the prefix of the other) and the class providing eventual prefix, in which the common prefix eventually converges. Interestingly, the authors of~\cite{ADLPT19} show that to solve strong prefix, Consensus is needed, however the work does not address the solvability of eventual prefix, which is the focus of this paper. 

 The objective of this paper is thus  to push  further this line of inquiry by presenting  
 an in-depth study of deterministic eventual finality.
We reinvestigate in Section~\ref{sec:definitions} the definition of (deterministic) eventual prefix consistency presented in~\cite{ADLPT19} to fit both the context in which an infinite number of blocks are appended to the blockchain.  We introduce the notion of bounded displacement, which informally says that  the number  of  blocks  that  can  be  revoked  from  the  current blockchain is bounded. Providing solutions that guarantee a known bound of the displacement reveals to be an important crux in the construction of blockchains. Specifically, in Section~\ref{sec:eventual_consensus_reductions} we  show that known bounded displacement eventual prefix consistency, that is eventual prefix consistency guaranteeing  a bounded and known displacement, is equivalent to Consensus. We also demonstrate  that unknown bounded displacement eventual prefix consistency, that is eventual prefix consistency guaranteeing  a bounded but unknown displacement, is equivalent to eventual strong prefix consistency. Finally we show that eventual strong prefix consistency is stronger than Eventual Consensus, an abstraction that captures eventual agreement among all participants. Then in Section~\ref{sec:eventual_prefix_solutions} we provide an algorithm  that guarantees eventual prefix consistency in an asynchronous environment with an unbounded number of Byzantine processes. We are not aware of any such solution in the literature. We also show that it is impossible to build a blockchain that guarantees eventual prefix consistency based on the longest chain rule. Note that such a selection function is used by many cryptoassets blockchains.  Finally, we discuss impossibilities and possibilities of unbounded displacement eventual prefix consistency. In particular, we propose an algorithm that solves unknown bounded displacement eventual prefix in an eventually synchronous environment in presence of less than a majority of Byzantine processes. Once again, we are not aware of any such solution in the literature. 
To summarize, we close the gap between eventual prefix and strong prefix consistencies.

\section{Definitions}
\label{sec:definitions}

\subsection{Preliminary Definitions}
Similarly to~\cite{ADLPT19}, we describe a blockchain object as an abstract data type which allows us to completely characterize a blockchain by the operations it exports~\cite{Liskov}.
The basic idea underlying the use of abstract data types is to specify shared objects using two
complementary facets: a sequential specification that describes the semantics of the object, and
a consistency criterion over concurrent histories, i.e. the set of admissible executions in a concurrent environment~\cite{matthieu-book}. Prior to presenting the blockchain abstract data type  we first recall the formalization used to describe an abstract data type (ADT).

\subparagraph*{Abstract data types.}
An abstract data type ($\textit{ADT}$) is a tuple of the form
$T=(A,B,Z,z_0,\tau,\delta)$.
Here~$A$ and $B$ are countable sets called the \emph{inputs} and
\emph{outputs}.  $Z$ is a countable set of abstract object 
\emph{states}, $z_0\in Z$ being the initial state of the object. 
The map~$\tau: Z\times A \to Z$ is the \emph{transition function}, specifying the effect of an input on the object state 
and the map~$\delta: Z\times A \to B$ is the \emph{output function}, specifying the output returned for a given input and 
object local state. The input represents an operation with its parameters, where (i) the operation can have a side-effect that changes the abstract state according to transition function $\tau$  and (ii) the operation can return values taken in the output $B$, which depends on the state in which it is called and the output function $\delta$.

\subparagraph*{Concurrent  histories of an ADT.}
\label{sec:conc-spec-adts}
Concurrent histories are defined considering asymmetric event structures, i.e., partial order relations among events executed by different processes.
\begin{definition}
\label{def:adthistory}{\bf(Concurrent history $H$)}
	The execution of a program that uses an abstract data type  T =$\langle A, B, Z, \xi_0, \tau, \delta \rangle $ defines a concurrent history  $H=\langle \Sigma, E, \Lambda, \mapsto, \prec, \nearrow \rangle$,  where
	\begin{itemize}
		\item $\Sigma = A \cup ( A \times B)$ is a countable set of operations;
		\item $E$ is a countable set of events that contains all the ADT operations invocations  and all ADT operation response events;
		\item $\Lambda:E \rightarrow \Sigma$ is a function which associates events to the operations in $\Sigma$;
		\item $\mapsto$: is the process order, irreflexive order over the events of $E$. Two events $(e,e') \in E^2$ are ordered by $\mapsto$ if they are produced by the same process, $e \neq e'$ and $e$ happens before $e'$, that is denoted as $e \mapsto e'$. 
		\item  $\prec$: is the operation order, irreflexive order over the events of $E$. For each couple $(e, e') \in E^2$ 
		if $e'$ is the invocation of an operation occurred at time $t'$ and $e$ is the response of another operation occurred at time $t$ with $t < t'$ then $e \prec e'$; 
		\item $\nearrow$: is the program order, irreflexive order over $E$, for each couple $(e, e') \in E^2$ with $e \neq e'$ if $e \mapsto e'$ or $e \prec e'$ then $e \nearrow e'$.	
\end{itemize}
\end{definition}

\subsection{The blockchain ADT}

Following~\cite{ADLPT19}, we represent a blockchain structure as a a tree of blocks. Indeed, while consensus-based blockchains prevent forks or branching in the tree of blocks, most blockchain systems based on proof-of-work allow the occurrence of forks to happen hence presenting blocks under a tree structure.  The blockchain object is thus defined as a  blocktree abstract data type (Blocktree ADT). However, we modify the original version of the blocktree ADT proposed in~\cite{ADLPT19} so that  distinct selection functions $f_r()$ and $f_a()$  for respectively  the {\sf read()} and {\sf append()} operations  exist. As will be shown latter in the paper, by relying on different functions  to select  the blockchain to be returned by a {\sf read()}  operation and to select the chain of the blocktree to which the new block is {\sf appended}, this allows us to present the first  algorithm that constructs a blockchain satisfying deterministic eventual finality  in an asynchronous system  in presence of an unbounded number of Byzantine processes. 

\subsubsection{Sequential Specification}
A blocktree data structure is a directed rooted tree $bt = (V_{bt},E_{bt})$ where $V_{bt}$ represents a set of blocks and $E_{bt}$ a set of edges such that each block has a single path towards the root of the tree $b_0$ called the genesis block. Let $\mathcal{B}$ be the countable and non empty set of uniquely identified blocks and let $\mathcal{BC}$ be the countable non empty set of blockchains, where a blockchain is a path from a leaf of $bt$ to $b_0$. A blockchain is denoted by $bc$. 
The structure is equipped with two operations {\sf append}() and  {\sf read}().
Operation  {\sf append}(b)  adds the block $b\not\in bt$ to $V_{bt}$ and adds the edge $(b,b')$ to $E_{bt}$ where $b'\in V_{bt}$ is returned by the append selection function $f_a()$  applied to $bt$.
Operation  {\sf read}()  returns  the chain $bc$ selected by the read selection function $f_r()$ applied to $bt$. In the following the chain $bc$ returned by a {\sf read}() operation $r$ is often denoted by $r/bc$.
Only blocks satisfying some validity predicate $P$ can be appended to the tree. The definition of the Blocktree (BT) ADT amended with the two types of selection functions is as follows.

\begin{definition}{(Blocktree ADT (BT-ADT)} The Blocktree Abstract Data Type is the $6$-tuple $\mathrm{BT-ADT}=\lbrace A=\{\mathit{append}(b),\mathit{read}()/bc \in \mathcal{B}\}, B = \mathcal{BC}\cup\{\top,\bot\},Z=\mathcal{BT},\xi_0=b_0,\tau,\delta\rbrace$, 
where the transition function $\tau:Z\times A \rightarrow Z$ is defined by 
\begin{align*}
	\tau(bt,\mathit{read}()) &= bt\\
	\tau(bt,\mathit{append(b)}) &=\left\{\begin{array}{l}
	(V_{bt}\cup \{b\},E_{bt}\cup\{b,\mathit{last\_block}(f_a(bt))\}) \mathrm{\ if\ } P(f_a(bt)^\frown b)\\
	bt \mathrm{\ otherwise,}
	\end{array}\right.
\end{align*}
and where the output function $\delta: Z\times A \rightarrow B$ is defined by
\begin{align*}
	\delta(bt,\mathit{read}()) &= f_r(bt)\\
	\delta(bt,\mathit{append(b)}) &=\left\{\begin{array}{l}
	\top \mathrm{\ if\ } P(f_a(bt)^\frown b)\\
	\bot \mathrm{\ otherwise.}
	\end{array}\right.
\end{align*}
\end{definition} 
The read selection $f_r()$ takes as argument the blocktree and returns a chain, that is a sequence of blocks starting from the genesis block to a leaf block of the blocktree. The chain returned by the ${\sf read()}$ operation is called the blockchain. The append selection function $f_a()$ takes as argument the blocktree and returns a chain of blocks. Function $\mathit{last\_block}()$ takes as argument a chain of blocks and  returns the last appended block of the chain.

Predicate $P$ is an application-dependent predicate used to verify the validity of the chain obtained by appending the new block $b$ to the chain returned by $f_a()$. In Bitcoin for instance this predicate embeds the logic to verify that the obtained chain does not contain double spending or overspending transactions. 

Note that we do not need to add the validity check during the {\sf read} operation in the sequential specification, because in absence of concurrency the validity check during the {\sf append} operation is enough.   

\subsubsection{Concurrent specification and consistency criteria}
\label{sec:conc-spec-bt}
The concurrent specification of a blocktree abstract data type is the set of concurrent histories. A blocktree consistency  criterion is a function that returns the set of concurrent histories admissible for a blocktree abstract data type.  We define three consistency criteria for the blocktree: \emph{BT Eventual Prefix  consistency}, \emph{BT Strong Prefix consistency} and \emph{BT Eventual Strong consistency}. For ease of readability, we employ the following notations:

\begin{itemize}
	\item $E(a^*,r^*)$ is an infinite set containing an infinite number of {\sf append}$()$ and {\sf read}$()$ invocation and response events;
	\item $E(a,r^*)$ is an infinite set containing {\em (i)} a finite number of {\sf append}$()$ invocation and response events and {\em (ii)} an infinite number of {\sf read}$()$ invocation and response events;
	\item $o_{inv}$ and $o_{rsp}$ indicate respectively the invocation and response event of an operation $o$; and in particular for the {\sf read}$()$ operation, $r_{rsp}/bc$ denotes the returned blockchain $bc$ associated with the response event $r_{rsp}$ and for the {\sf append}$()$ operation $a_{inv}((b))$ denotes the invocation of the append operation having $b$ as input parameter; 
	\item ${\sf length}: \mathcal{BC}\rightarrow \mathbb{N}$ denotes a monotonic increasing deterministic function that takes as input a blockchain $bc$ and returns a natural number as length of $bc$. 
	Increasing monotonicity means that ${\sf length}(bc^\frown \{b\}) > {\sf length}(bc)$; 
	\item $bc \sqsubseteq bc^\prime$ iff $bc$ prefixes $bc^\prime$.
	\item $bc[i]$ refers to the $i-th$ block in the blockchain $bc$.
\end{itemize}

\subsubsection*{Eventual Prefix (EP) Consistency}

\begin{definition}[BT Eventual  Prefix Consistency (EP) criterion]
	A concurrent history $H=\langle \Sigma, E, \Lambda, \mapsto, \prec, \nearrow \rangle$, of a system that uses a BT-ADT, verifies the BT Eventual Prefix Consistency criterion if the following  properties hold:
	\begin{itemize}
\item {\bf Chain validity:} 
		
		 $\forall r_{rsp} \in E, P(r_{rsp}/bc)$.	
		 
		\textit{ Each returned chain is valid}.
\item {\bf Chain integrity:} 
		
		 $\forall r_{rsp} \in E,  \forall b \in r_{rsp}/bc: b\neq b_0, \exists a_{inv}(b) \in E, a_{inv}(b) \nearrow r_{rsp}$.	
	
	\textit{If a block different from the genesis block is returned, then an {\sf append} operation has been invoked with this block as parameter. This property is to avoid the situation in which {\sf reads} return blocks never {\sf appended}.}
\item {\bf Eventual prefix: }

 $\forall E\in E(a,r^*)\cup E(a^*,r^*),\forall r_{rsp}/bc, \forall i\in \mathbb{N}: bc[i]\neq \bot, \exists r'_{rsp}, \forall r''_{rsp}:  r'_{rsp} \nearrow r''_{rsp},  ((r'_{rsp}/bc)[i] = (r''_{rsp}/bc)[i])$.

\textit{In all the histories in which the number of {\sf read} invocations is infinite,  then for any non empty read chain position $i$, there exists a {\sf read} $r'/bc'$ from which all the subsequent {\sf reads} $r''/bc''$  will return the same block at position $i$, i.e. $bc'[i]=bc''[i]$.}
	\item {\bf Ever growing tree:} 

$\forall E\in E(a^*,r^*), \forall k\in\mathbb{N}, \exists r\in E:
\mathbf{length}(r_\mathit{rsp}/bc)> k.$
 
\textit{In all the histories in which the number of {\sf append} and {\sf read} invocations is infinite, for each length $k$, there exists a {\sf read} that returns a chain with length greater than $k$.
 This property avoids the trivial scenario in which the length of the chain remains unchanged despite the occurrence of an infinite number of {\sf append} operations. This can happen for instance if the tree is built as a star with infinite branches of bounded length.}
	\end{itemize}
\end{definition}

\subsubsection*{Strong  Prefix (SP) Consistency}

\begin{definition}[BT Strong  Prefix Consistency criterion ($SP$)]
	A concurrent history $H=\langle \Sigma, E, \Lambda, \mapsto, \prec, \nearrow \rangle$ of the system that uses a BT-ADT verifies the BT Strong Consistency criterion if Chain validity, Chain integrity, Ever growing tree (as defined for EP)  and the following property hold:
	\begin{itemize}

	\item {\bf Strong prefix: } 
	
		$\forall r_{rsp}, r'_{rsp} \in E^2,  (r'_{rsp}/bc' \sqsubseteq r_{rsp}/bc) \vee (r_{rsp}/bc \sqsubseteq r'_{rsp}/bc')$.
		
		\textit{For each pair of returned blockchains, one blockchain is the prefix of the other.}
		
	\end{itemize}
\end{definition}

\subsubsection*{Eventual Strong Prefix (ESP) Consistency}
\begin{definition}[BT Eventual Strong Prefix Consistency criterion ($ESP$)]
	A concurrent history $H=\langle \Sigma, E, \Lambda, \mapsto, \prec, \nearrow \rangle$ of the system that uses a BT-ADT verifies the BT Strong Consistency criterion if Chain validity, Chain integrity, Ever growing tree (as defined for EP)  and the following property hold:
	\begin{itemize}
	
	\item {\bf Eventual strong prefix: } 
	
 $\forall E\in E(a,r^*)\cup E(a^*,r^*)$, $\exists  r_{rsp}\in E, \forall r'_{rsp}, r''_{rsp} \in E^2: r_{rsp} \nearrow r'_{rsp}\wedge r_{rsp} \nearrow r''_{rsp},  (r''_{rsp}/bc' \sqsubseteq r'_{rsp}/bc) \vee (r'_{rsp}/bc \sqsubseteq r''_{rsp}/bc')$.
		
		\textit{In all histories with an infinite number of {\sf reads}, there exists a {\sf read} $r$ from which  for each pair of returned blockchains, one blockchain is the prefix of the other.}
		
	\end{itemize}
\end{definition}

\subsubsection*{Bounded  displacement} 
 Informally, the  bounded displacement says that for any two {\sf reads} $r/bc$ and $r'/bc'$ such that $r$ precedes $r'$, then by pruning the last  \emph{dis} blocks from bc the obtained  chain is a prefix of bc'. Note that constant \emph{dis} can be initially known or not.
 
\begin{definition}{Bounded displacement}
\begin{itemize}
    \item $\exists {\sf dis}\in\mathbb{N},\forall r_{rsp}, r'_{rsp}\in E: r_{rsp} \nearrow r'_{rsp}, \forall i\in\mathbb{N}: i\leq\mathbf{length}(r_{rsp}/bc)- {\sf dis}, (r_{rsp}/bc)[i]=(r'_{rsp}/bc')[i]$.\\
   
\end{itemize}

\end{definition}

We will show in the following that satisfying  the eventual prefix consistency criterion plus the bounded displacement property with known $dis$ boils down to getting  the strong prefix consistency criterion. This is because any algorithm implementing the read selection function $f_r$, can safely select the current chain but the last $dis$ blocks. Eventual prefix consistency criterion plus the known bounded displacement is designated in the sequel by \emph{known bounded displacement  eventual prefix consistency}. 

The eventual prefix consistency criterion plus the bounded displacement property with unknown $dis$ boils down to get the eventual strong prefix consistency criterion. In this case a simple algorithm implementing the {\sf read} operation always returns half of the selected longest chain. Since chains are always growing, the number of removed blocks increases up to reaching $dis$. 
In the sequel, eventual prefix consistency criterion plus the unknown bounded displacement is designated by \emph{unknown bounded displacement  eventual prefix consistency}. 

In the following section we prove the above-mentioned equivalences more formally and study relationships to known problems such as Consensus and Eventual Consensus to determine the assumptions on the system needed to implement blocktree consistency models.

\section{(Eventual) Consensus Reductions}
\label{sec:eventual_consensus_reductions}
In this Section we investigate the impact of the bounded displacement property combined with the eventual prefix consistency problems. In particular, we derive that when the displacement is known, such problem is equivalent to Consensus, while when unknown, this problem is stronger than Eventual Consensus~\cite{DGKPS15}.

\subsection{Known Bounded Displacement and Consensus}
\begin{restatable}{theorem}{eQEPKBandC}\label{th:eQEPKBandC}
Known bounded displacement eventual prefix  is equivalent to Consensus.
\end{restatable}

\begin{proof}
We first show how to solve strong prefix given a solution $\mathcal{P}_{EP}$ for known bounded displacement eventual prefix and then the reciprocal direction.  
Indeed, the equivalence between the strong prefix consistency criteria with Consensus is known from~\cite{btadt}.

Let us show that we can solve the strong prefix consistency criteria, known to be equivalent to Consensus~\cite{btadt}, using  $\mathcal{P}_{EP}$. 
To do so, we consider the following transformation from  the protocol $\mathcal{P}_{EP}$. To make an {\sf append}$()$ operation, processes simply use the {\sf append}$()$ operation provided by $\mathcal{P}_{EP}$. 
But, for the the {\sf read}$()$ operation, processes use the {\sf read}() operation provided by $\mathcal{P}_{EP}$ to obtain a chain and prune the last $dis$ blocks from it before returning the remaining chain. Note that if there are less
than $dis$ blocks, processes then return the genesis block.

Let us show that this modified protocol solves the  strong prefix consistency. For this, we need to show that the following properties are satisfied:
\begin{itemize}
    \item {\bf Chain validity:} The chain validity property is still satisfied by pruning $dis$ blocks.
    \item {\bf Chain integrity:} The chain integrity property is still satisfied by pruning $dis$ blocks.
    \item {\bf Strong prefix: } The  strong prefix property follows from the known bounded displacement and the removal of the last $dis$ blocks. Indeed, if we remove the last $dis$ blocks, then for any two {\sf read}$()$ operations, then the first {\sf read}$()$ returns a prefix of the second {\sf read}$()$ operation. 
	\item {\bf Ever growing tree:} The ever growing tree property is still satisfied by pruning $dis$ blocks.
\end{itemize}

For the other direction, we can build a solution to known bounded displacement eventual prefix using a solution for the strong prefix consistency criteria. This trivially solves the eventual known bounded displacement prefix property with $dis=0$.
\end{proof}

If we consider known bounded displacement eventual prefix then we have strong prefix (we can read a chain and remove the last $dis$ blocks to be sure that what we read is final).
Thus, since strong prefix requires Consensus and Consensus in impossible in an asynchronous system (with a single failure) then, known bounded displacement eventual prefix is impossible to solve in an asynchronous system.

\subsection{Unknown Bounded Displacement and Eventual Consensus}

\paragraph*{Eventual Strong prefix and Unknown Bounded Displacement.}
If we have an unknown bounded displacement, then eventual prefix, combined with such property is stronger than Eventual Consensus. To show that, we first show its equivalence with eventual strong prefix. Later we recall the Eventual Consensus problem (with a small modification of the Validity property, to make it suits to the blockchain context) and finally we show that unknown bounded displacement eventual prefix is stronger than Eventual Consensus.

\begin{restatable}{theorem}{UBDEPandESP}
Unknown bounded displacement eventual prefix is equivalent to eventual strong prefix.
\end{restatable}

\begin{proof}
Let $\mathcal{P}$ be a protocol solving the unknown bounded displacement eventual prefix and let us show that we can solve the eventual strong prefix property. 
To do so, we consider the following modification to the protocol $\mathcal{P}$. To make an {\sf append}$()$ operation, processes simply use the {\sf append}$()$ operation provided by $\mathcal{P}$. 
But, for {\sf read}$()$ operation, processes use the {\sf read}$()$ operation to obtain a chain and prune the second half of the returned chain before returning the remaining half of the chain. 

Let us show that this modified protocol solves the eventual strong prefix property. For this, we need to show that the following properties are satisfied:
\begin{itemize}
    \item {\bf Chain validity:} The chain validity property is still satisfied by pruning half the chain.
    \item {\bf Chain integrity:} The chain integrity property is still satisfied by pruning half the chain.
    \item {\bf Eventual Strong prefix: } The eventual strong prefix property follows from the unknown bounded displacement and the removal of the second half of the chain. Indeed, if we remove the second half of the chain, then eventually for any two {\sf read}$()$ operations, then the first {\sf read}$()$ returns a prefix of the second {\sf read}$()$ operation. Indeed, since we remove a growing number of blocks, eventually we remove more than $dis$ blocks and obtain chains such that one is the prefix of the other. 
	\item {\bf Ever growing tree:} The ever growing tree property is still satisfied by pruning half the chain.
\end{itemize}

For the other direction, let us consider a protocol $\mathcal{P}$ solving the eventual strong prefix property and let us show that it solves the unknown bounded displacement eventual prefix. The property of eventual strong prefix clearly implies the eventual prefix property. 
Let ${\sf displacement}(b_1, b_2)$ be the function that takes two blockchains $b_1$ and $b_2$ and returns the number of blocks needed to prune $b_1$ and obtain $b_1^\prime$ such that $b_1 ^\prime \sqsubseteq b_2$. 
Let us show that $\exists dis\in\mathbb{N},\forall r_{rsp},r_{rsp}'\in E^2,r\nearrow r', {\sf displacement}(r_{rsp}/bc,r_{rsp}'/bc)<dis$. 
Assume by contradiction that this property is not satisfied, then it implies that for any $dis$, there exists a couple of reads with a greater displacement than $dis$. This implies that the eventual strong prefix property is not satisfied which leads to a contradiction, hence eventual strong prefix property implies unknown bounded displacement eventual prefix.
\end{proof}

%%%%%%%%%%%%%%%%%%%%%%%%%%%%%%%%%%%%
\begin{comment}
\subsubsection{Eventual Consensus}
The eventual consensus~\cite{DGKPS15} (EC) abstraction exports, to every process $p_i$, operations ${\sf proposeEC_1}, {\sf proposeEC_2}, \dots$ that take binary arguments and return binary responses. Assuming that, for all $j \in \mathbb{N}$, every process invokes ${\sf proposeEC_j}$ as soon as it returns a response to ${\sf proposeEC_{j-1}}$, the abstraction guarantees that, in every admissible run, there exists $k \in \mathbb{N}$, such that the following properties are satisfied:
\begin{itemize}
    \item {\bf EC-Termination} Every correct process eventually returns a response to ${\sf proposeEC_j}$ for all $j \in \mathbb{N}$.
    \item {\bf EC-Integrity} No process responds twice to ${\sf proposeEC_j}$ for all $j \in \mathbb{N}$.
    \item {\bf EC-Validity} Every value returned to ${\sf proposeEC_j}$ was previously proposed to ${\sf proposeEC_j}$ for all $j \in \mathbb{N}$.
    \item {\bf EC-Agreement} No two processes return different values to ${\sf proposeEC_j}$ for all $j \geq k$.
\end{itemize}

It is straightforward to transform the binary version of EC into a multi-valued one with unbounded set of inputs [24], thus we directly refer to EC multi-valued version.
\end{comment}
%%%%%%%%%%%%%%%%%%%%%%%%%%%%%%%%%%%%

\paragraph*{Eventual Consensus.}
The eventual consensus~\cite{DGKPS15} (EC) abstraction captures eventual agreement among all participants. It exports, to every process $p_i$, operations ${\sf proposeEC_1}, {\sf proposeEC_2}, \dots$ that take multi-valued arguments (correct processes propose valid values) and return multi-valued responses.  Assuming that, for all $j \in \mathbb{N}$, every process invokes ${\sf proposeEC_j}$ as soon as it returns a response to ${\sf proposeEC_{j-1}}$, the abstraction guarantees that, in every admissible run, there exists $k \in \mathbb{N}$ and a predicate $P_{EC}()$, such that the following properties are satisfied:
\begin{itemize}
    \item {\bf EC-Termination} Every correct process eventually returns a response to ${\sf proposeEC_j}$ for all $j \in \mathbb{N}$.
    \item {\bf EC-Integrity} No process responds twice to ${\sf proposeEC_j}$ for all $j \in \mathbb{N}$.
    \item {\bf EC-Validity} Every value returned to ${\sf proposeEC_j}$ is valid with respect to $P_{EC}()$.
    %\item {\bf EC-Validity2} If all correct processes propose the same value to ${\sf proposeEC_j}$, then this value is returned by ${\sf proposeEC_j}$ for all $j \in \mathbb{N}$.
    \item {\bf EC-Agreement} No two correct processes return different values to ${\sf proposeEC_j}$ for all $j \geq k$.
\end{itemize}

\begin{restatable}{theorem}{eQEPUBandEC}\label{th:eQEPUBandEC}
Eventual strong prefix is stronger than Eventual Consensus.
\end{restatable}

\begin{proof}
We show that it exists a protocol $\mathcal{P_{EC}}$ to solve Eventual Consensus  starting from a protocol $\mathcal{P_{ESP}}$ for eventual strong prefix. We do the transformation as follows. Every correct process $p$ invokes ${\sf proposeEC_j}$ for all $j \in \mathbb{N}$. We impose that the validity predicate $P()$ of the blocktree ADT (see Section~\ref{sec:definitions}) be equal to predicate $P_{EC}()$. When a correct process invokes the ${\sf proposeEC_j}(v)$ operation of $\mathcal{P_{EC}}$, for any $j \in \mathbb{N}$, then it directly invokes the {\sf append}$(v)$ operation of $\mathcal{P_{ESP}}$. After that, $p$ invokes a sequence of ${\sf read}()$ operations up to the moment it returns a chain $bc$ such that $bc[j] \neq \bot$. Process $p$  then returns chain $bc$  as decision for ${\sf proposeEC_j}(v)$ and triggers the next operation ${\sf proposeEC_{j+1}}(v')$. 

Let us show that  protocol  $\mathcal{P_{EC}}$ solves the Eventual Consensus.
\begin{itemize}
    \item {\bf EC-Termination} This property is guaranteed by the ever growing tree property.
    \item {\bf EC-Integrity} This property follows directly from the transformation.
    \item {\bf EC-Validity} This property follows by construction and the chain validity property, since $P()=P_{EC}()$.
    \item {\bf EC-Agreement} This property follows by the eventual strong prefix property, that guarantees that there exists a ${\sf read}()$ operation $r$ such that, all the subsequent ones, return blockchains that are each prefix of the other one. In other words, eventually there is agreement on the value contained in $bc[j]$. This implies that there exists $k$ for which all ${\sf proposeEC_j}$ with $j>k$ returns the same value to all correct processes.
\end{itemize}
 \end{proof}

\section{Eventual Prefix solutions}
\label{sec:eventual_prefix_solutions}
In this section we assume that processes are equipped with an Oracle $\Theta_P$ or $\Theta_{F,k}$ (as defined in~\cite{btadt}) to access the blockchain, i.e., to read its status ($\mathit{getValidBlock}()$ operation) and append a new block to it ($\mathit{setValidBlock}()$ operation). Briefly, oracles are abstraction that encapsulate the creation of valid blocks and since only valid blocks can be appended to the blocktree, it follows that
oracles grant the access to the blocktree. Oracle  $\Theta_{F,k}$  also owns a synchronization
power to control the number of forks ($k$), in terms of branches of the blocktree from a given block. On the other hand $\Theta_P$ poses no constraints on the number of blocks that can be appended to the same block, i.e., it allows forks and it poses no bound on the number of blocks in a fork (more details can be found in in~\cite{btadt}).

%\textcolor{blue}{I wonder whether it would be useful to indicate that, in the following two sections, we are not providing algorithms that describe the implementation of the blocktree  within a distributed system in which processes communicate by a message passing system but we provide  algorithms that implement the read and write operations of the shared blocktree object with the sought properties. What do you think ?}
%~\\

\subsection{Impossibility of Eventual Prefix with longest chain rule}

%One of the central challenges is to define the selection functions. 
%We have two functions to define, the read selection function and mostly the append selection function. 

In the following we prove that we cannot provide eventual prefix consistency in the shared-memory model if the append selection function follows the the longest chain rule even when equipped with $\Theta_{F,2}$.

We consider a crash-prone asynchronous system with a finite number of processes and an append selection function $f$ equipped with $\Theta_{F,2}$ that deterministically selects a chain with the longest chain, and in case of a tie  selects a chain based on the lexicographical order. We show that it is impossible to implement a blockchain that satisfies  eventual prefix consistency for this append selection function. 

Note that such a selection function is used by many blockchain systems. In particular in proof-of-work systems such as Bitcoin, chains are selected as the longest chain while in Ethereum chains are selected using the chain with greatest weight, both captured by the selection of chains according to the longest chain. 

\tikzset{
	mystyle/.style={
		circle,
		inner sep=1pt,
		minimum size=1cm,
		thick,
		align=center,
		draw=black,
	}
}

\tikzset{
	mystyle2/.style={
		draw=white,
	}
}

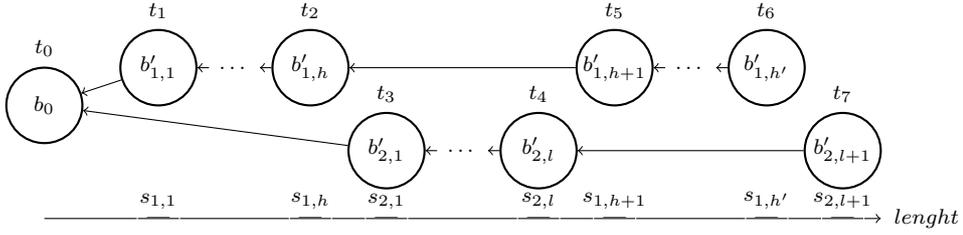
\begin{figure}
	\begin{tikzpicture}	
	\node [mystyle] (b0) [label=\scriptsize $t_0$] at (3,3.5) {\scriptsize $b_0$};
	
	\node [mystyle] (p11) [label=\scriptsize $t_1$] at (4.5,4) {\scriptsize $b^\prime_{1,1}$};
	\node [mystyle2] (p11dots) at (5.5,4) {\scriptsize $\dots$};	
	\node [mystyle] (p12) [label=\scriptsize $t_2$] at (6.5,4) {\scriptsize $b^\prime_{1,h}$};
	\node [mystyle] (p13) [label=\scriptsize $t_5$] at (10.5,4) {\scriptsize $b^\prime_{1,h+1}$};
	\node [mystyle2] (p13dots) at (11.5,4) {\scriptsize $\dots$};	
	\node [mystyle] (p14) [label=\scriptsize $t_6$] at (12.5,4) {\scriptsize $b^\prime_{1,h'}$};
	
	\node [mystyle] (p21) [label=\scriptsize $t_3$] at (7.5,2.9) {\scriptsize $b^\prime_{2,1}$};
	\node [mystyle2] (p21dots) at (8.5,2.9) {\scriptsize $\dots$};
	\node [mystyle] (p22) [label=\scriptsize $t_4$] at (9.5,2.9) {\scriptsize $b^\prime_{2,l}$};
	\node [mystyle] (p23) [label=\scriptsize $t_7$] at (13.5,2.9) {\scriptsize $b^\prime_{2,l+1}$};

	\path[<-] (b0) edge (p11) (p11) edge (p11dots)  (p11dots) edge (p12) (p12) edge (p13) (p13) edge (p13dots) (p13dots) edge (p14);	
	\draw[<-] (b0) edge (p21)  (p21)edge (p21dots) (p21dots) edge (p22)  (p22)edge(p23);

    \draw[->] (3, 2) -- (14, 2); \node[mystyle2] at (14.6,2) {\scriptsize$lenght$};
    \node[mystyle2] () [label={[yshift=-.2cm]\scriptsize $s_{1,1}$}] at (4.5, 2) {\scriptsize |};
	\node[mystyle2] () [label={[yshift=-.2cm]\scriptsize $s_{1,h}$}] at (6.5, 2) {\scriptsize |};
	\node[mystyle2] () [label={[yshift=-.2cm]\scriptsize $s_{2,1}$}] at (7.5, 2) {\scriptsize |};
	\node[mystyle2] () [label={[yshift=-.2cm]\scriptsize $s_{2,l}$}] at (9.5, 2) {\scriptsize |};
	\node[mystyle2] () [label={[yshift=-.2cm]\scriptsize $s_{1,h+1}$}] at (10.5, 2) {\scriptsize |};
	\node[mystyle2] () [label={[yshift=-.2cm]\scriptsize $s_{1,h'}$}] at (12.5, 2) {\scriptsize |};
	\node[mystyle2] () [label={[yshift=-.2cm]\scriptsize $s_{2,l+1}$}] at (13.5, 2) {\scriptsize |};
	\end{tikzpicture}
	\caption{\label{fig:proof}A blocktree generated by two processes. On the x-axis the longest chain value of each chain at different time instants (from the root to the current leaf) and the relationships between those values. %Each read operation performed at a labelled time instant returns the corresponding branch, e.g., from $b_0$ to $b_h$ at time $t_2$ and from $b_0$ to $b_j$ at time $t_4$.
	}
\end{figure}

\begin{theorem}
It is impossible to build a blockchain that guarantees eventual prefix consistency when the {\sf append} operation is based on the longest chain rule and is equipped with $\Theta_{F,2}$.
\end{theorem}
\begin{proof}
Intuitively, the impossibility follows from the fact that with the longest chain selection, races can occur between different branches in the tree. We show that as forks may occur, we can create two infinite branches sharing only the root. Alternatively, one or the other branch constitutes the longest chain and {\sf append} operations selects chains from each branch alternatively. This is enough to show that the only common prefix that is returned is the root, hence, violating the eventual prefix consistency. 

Let $p_1$ and $p_2$ be two processes trying to {\sf append} infinitely many blocks. 

At time~$t_0$, for both~$p_1$ and~$p_2$, the $\mathit{update\_view}$ of $bt$ equals $b_0$, thus both call $\mathit{getValidBlock}(b_0, b_{i,1})=b^\prime_i$, where $i=1$ for $p_1$ and $i=2$ for $p_2$. At time~$t_1 > t_0$, $p_1$ and $p_2$ are poised to call $\mathit{setValidBlock}(b_0, b^\prime_{i,1})$. Let $p_2$ be slower than $p_1$, thus $p_1$ cannot distinguish whether $p_2$ is slow or crashed. Process~$p_1$ gets the set $\{b^\prime_{1,1}\}$ as result of the function and proceeds to {\sf append}  a new block $b_{1,2}$, i.e., it updates $bt$'s view and applies $f$ on it to select the leaf where to {\sf append}  the new block, in this case the only leaf is $b^\prime_{1,1}$. $p_1$ calls $\mathit{getValidBlock}(b^\prime_{1,1}, b_{1,2})$ which returns $\{b^\prime_{1,2}\}$ and hereafter at time $t_2>t_1$ it is poised to call $\mathit{setValidBlock}(b^\prime_{1,1}, b^\prime_{1,2})$. 

We can continue in the same way up to time $t_3 \geq t_2$, the moment in which the function $\mathit{setValidBlock}(b_0, b^\prime_{2,1})$ returns $\{b^\prime_{1,1},b^\prime_{2,1}\}$ to~$p_2$. We assume that at time $t_3$, $p_1$ built a chain from the root $b_0$ to $b^\prime_{1,h}$ and is poised to {\sf append}  $b^\prime_{1,h+1}$. The chain from $b_0$ to $b^\prime_{1,h}$ has a length $s_{1,h}$ and the chain $b_0$ to $b^\prime_{2,1}$ has a length~$s_{2,1}$. W.l.o.g., we assume that $s_{1,h}<s_{2,1}$ and~$s_{1,h+1}>s_{2,1}$. At time $t_4>t_3$, $p_2$ updates its view of $bt$ and $f$ applied to it selects the chain from $b_0$ to $b^\prime_{2,1}$, as $s_{1,h}<s_{2,1}$, to {\sf append}  the new block. We can make it {\sf append}  blocks until it is poised to call a $\mathit{setValidBlock}(b^\prime_{2,l}, b^\prime_{2,l+1})$ such that $s_{2,l}<s_{1,h+1}$ and~$s_{2,l+1}>s_{2,h+1}$ (cf. Figure~\ref{fig:proof}). 
We can continue to create two infinite branches sharing only the root. A {\sf read} operation can alternatively return a chain from each branch, depending on the time (cf. Figure~\ref{fig:proof}), e.g., from $b_0$ to $b^\prime_{1,h}$ at time $t_2$ and from $b_0$ to $b^\prime_{2,l}$ at time $t_4$. Thus, the common prefix never increases, and so,  eventual prefix consistency is not  satisfied.
\end{proof}

\subsection{Asynchronous solution to eventual prefix with an unbounded number of Byzantine processes}
We consider an asynchronous system with a possibly infinite set of processes which can append infinitely many blocks, and processes can be affected by Byzantine failures. Each process has a unique identifier that can be verified using signatures. We show that in that setting it is possible to build a blockchain that satisfies eventual prefix consistency. This is done by using the oracle $\Theta_P$ which emulates a shared-memory model.

\begin{algorithm}
\DontPrintSemicolon % Some LaTeX compilers require you to use \dontprintsemicolon instead
\KwIn{Blockchain \textit{bt}}
$\mathit{bt}.\mathit{update\_view}()$\;
\If{$\Theta_P.\mathit{getValidBlock}(f_a(bt))$}{
	$\Theta_P.\mathit{setValidBlock}(f_a(bt))$\;
\Return{$\top$}
}
\Else{
\Return{$\bot$}
}
\caption{Append procedure}
\end{algorithm}

\begin{algorithm}
\DontPrintSemicolon % Some LaTeX compilers require you to use \dontprintsemicolon instead
\KwIn{Blockchain \textit{bt}}
$\mathit{bt}.\mathit{update\_view}()$\;
\Return{$f_r(\mathit{bt})$}
\caption{Read procedure}
\end{algorithm}

\paragraph*{Algorithm}

The main idea of the algorithm consists on using a local selection function for append operations.

The append select function $f_r$ selects a chain by going from the root to a leaf, choosing at each fork the edge to the child with the lowest identifier, given some arbitrary total order on blocks. 
The read selection function $f_r$ is the same as $f_a$.

To perform an {\sf append()} operation, processes update the current view of the $bt$'s state  and apply the $\mathit{getValidBlock}$ operation of $\Theta_P$ for the last block in $f_a(bt)$. Finally, they apply the $\mathit{setValidBlock}$ to append it to $bt$.  For {\sf read()} operations, processes return the chain selected by $f_r$ on the  $bt$ returned by $update\_view $.

\paragraph*{Proof sketch}
Chain validity and chain integrity properties are satisfied by Oracle $\Theta_P$. The ever-growing tree property relies on the fact that each fork has a finite number of blocks since there are finitely many processes and each (Byzantine or correct) process can contribute with at most one block per parent as multiple children created by the same process are ignored by $f_a$. Thus, eventually, new blocks contribute to the growth of the tree. The eventual prefix property follows from the definition of $f_a$. 
For any $b$ in $bt$, let $t_b$ be the time after which no process contributes to the same block $b$. All correct processes trigger the append operation on chains containing the non-ignored child with the longest chain. By induction, it creates an ever-growing prefix selected by $f_a$. As read chains are selected according to the append selection function, it provides the eventual prefix property.

\subsection{Eventual synchronous solution for Unknown Bounded Displacement Eventual Prefix}
In this Section we discuss the impossibilities and possibilities of unknown bounded displacement eventual prefix consistency. In particular, we show that it is unsolvable in an asynchronous shared-memory or message-passing system if at least one process fails with Byzantine failures. Interestingly, if we consider the previous result with known bounded displacement, the impossibility holds even in an eventually synchronous system model with more than one third Byzantine processes (with not surprise, given the equivalence with the Consensus problem, see Theorem~\ref{th:eQEPKBandC}). Finally we prove that unknown bounded displacement eventual prefix consistency is solvable in an eventual synchronous message-passing system with a majority of correct processes. 

\begin{theorem}
There does not exist any solution that solves unknown bounded displacement eventual prefix in an asynchronous system with at least one Byzantine process.
\end{theorem}

\begin{proof}
The proof follows from the relationship between the unknown bounded displacement eventual prefix which is stronger than the Eventual Consensus problem (cf. Theorem \ref{th:eQEPUBandEC}), which is equivalent to the Leader Election problem \cite{DGKPS15} which cannot be solved in an asynchronous system with at least one Byzantine processes~\cite{R07}.
\end{proof}
%{\color{red} to manage the equivalence between EC original a ours.}

\begin{theorem}
There does not exist any solution that solves known bounded displacement eventual prefix in an eventual synchronous system with more than one third of Byzantine faulty processes.
\end{theorem}

\begin{proof}
The proof follows from the equivalence between known bounded displacement eventual prefix and Consensus (cf. Theorem \ref{th:eQEPKBandC}), which is unsolvable in a synchronous (and thus also in an eventually synchronous) system with more than one third of Byzantine faulty process~\cite{LSP82}.
\end{proof}
%so i do not understand the following theorem since it seems in contradiction with Th6
%The difference is on the fact that the displacement is known or unknown. In th6 it is known. So it's harder somehow. In the case of th7 the displacement does exist but it's unknown. and also changes the assmption on the number of faulty byzantines. in th7 we ask just for a majority of correct processes. 

%I am sorry but i still do not understand. ah non I see the point
%yes, known dis means that we are solving consensus.
%otherwise we can solve it with just a majority of correct processes. 

%faulty processes are crash processes ?
%Indeed, we have Byzantine, a minority of Byzantines
%Sorry i need time to think about it
%Sure ok
\begin{theorem}\label{thm:UBDEPsolution}
There exists a solution that solves unknown bounded displacement eventual prefix in an eventual synchronous system with a majority of correct processes.
\end{theorem}

In order to prove the existence of a solution, we start from an existing solution, Streamlet \cite{streamlet}, that guarantees strong prefix consistency under the assumption of less than a third of byzantine processes and eventual synchrony with a known message delay $\Delta$. We weaken both of these assumptions to provide a solution to the eventual strong prefix consistency criteria. In particular, we assume only a majority of correct processes, we do not explicitly use $\Delta$ and consider a slightly modified version of the protocol. In the following we first describe Streamlet and the discuss the modifications before providing the proof.

{\bf Streamlet protocol.} The Streamlet protocol works in a partially synchronous system with a known message delay $\Delta$ and a finite set of $n$ processes. In particular, before the Global Stabilisation Time (GST), message delays can be arbitrary; however, after GST, messages sent by correct processes are guaranteed to be received by correct processes within $\Delta$ time. The following assumption holds: $\Delta$-bounded assumption during periods of synchrony, when an honest node sends a message at time $t$, an honest recipient is guaranteed to receive it by time $max(GST, t + \Delta)$. \footnote{Notice that, in Streamlet \cite{streamlet} there is not the notion of time but of round, which denotes a basic unit of time.}

In Streamlet~\cite{streamlet},  each  epoch, composed of $2\Delta$, has  a  designated  leader  chosen  at  random  by  a publicly known hash function. The protocol works as follows:
\begin{itemize}
    \item {\bf Propose-Vote.} In every epoch: %{\color{red} maybe we should say something on how processes change epoch and how $\Delta$ is explicitly used}
    \begin{itemize}
        \item The epoch’s designated leader proposes a new block extending from the longest notarized chain it has seen (if there are multiple,  break ties arbitrarily).  The notion “notarized” is defined  below.
        \item Every  process  votes  for  the  first  proposal  they  see  from  the  epoch’s  leader,  as  long  as  the proposed block extends from (one of) the longest notarized chain(s) that the voter has seen. A vote is a signature on the proposed block.
        \item When a block gains votes from at least 2n/3 distinct processes, it becomes notarized.  A chain is notarized if its constituent blocks are all notarized.
    \end{itemize}
    \item {\bf Finalize.} Notarized does not mean final.  If in any notarized chain, there are three adjacent blocks with consecutive epoch numbers,  the prefix of the chain up to the second of the three blocks is considered final.  When a block becomes final, all of its prefix must be final too.
\end{itemize}
%Streamlet solves the strong prefix consistency criteria under the assumption of less than a third of byzantine processes and eventual synchrony with a known message delay $\Delta$. We are interested in weakening both of these assumption to provide a solution to the eventual strong prefix consistency criteria. For this we assume only a majority of correct processes and modify slightly the protocol. 
%The modifications are twofold.
We propose protocol $\mathcal{S}$ with  the following modifications to Streamlet. First, we only require that a block gains votes from a majority of distinct processes to become notarized, which means that forks can occur. The second modification goes deeper: if a fork occurs, then it is possible to detect Byzantine processes and to exclude them from the voters. This is done as follows. When, two conflicting chains are finalized, that is two finalized chains that are not the prefix of one another, then processes look for inconsistent blocks. That is, two notarized blocks $b, b^\prime$ are inconsistent  with one another if one of the following two conditions hold:
\begin{itemize}
    \item{Cond. 1} $b$ and $b^\prime$ share the same epoch.
    \item{Cond. 2} $b$ (resp. $b^\prime$) that corresponds to a higher epoch, $b.\mathit{epoch}<b^\prime.\mathit{epoch}$ (resp. $b^\prime.\mathit{epoch}<b.\mathit{epoch}$), with a strictly smaller height, $b^\prime.\mathit{height}<b.\mathit{height}$ (resp. $b.\mathit{height}<b^\prime.\mathit{height}$), than the other block.
\end{itemize}

If a process votes for blocks inconsistent with one another  then it is detected as Byzantine. 

\begin{proof}
Let us show that $\mathcal{S}$ is a solution for eventual strong prefix. Let us first show that when a fork occurs, then we detect at least a Byzantine process. For this let us show that it implies the existence of two inconsistent blocks and that the intersection of voters of the two blocks is not empty as two majority intersect. In the following we show that voting for two inconsistent blocks is a Byzantine failure. 

Let us first show that voting for two inconsistent blocks $b$ and $b'$ is a Byzantine failure. If the two blocks are inconsistent for Cond. 1, then the intersecting voters are Byzantine as correct processes vote only once per epoch. It follows that if a process $q$ votes for $b$ and $b'$ then $q$ is Byzantine.
If the two blocks are inconsistent for Cond. 2, then the intersecting voters are Byzantine as correct processes vote only for blocks extending one of the longest notarized chains. That is, if a correct process $p$ votes for $b$ it means that $b$ is extending a notarized block $b_{pred}$ that is of height $b.\mathit{height}-1$, therefore $p$ cannot vote for a block $b'$ later on with a height strictly smaller than $b.\mathit{height}$ because it needs to extend one of the longest notarized chain. It follows that if a process $q$ votes for $b$ and $b'$ then $q$ is Byzantine.  

Let us now show that when a fork occurs we must have two inconsistent blocks. Indeed, if there is a fork then we have two sequences of three adjacent blocks with consecutive epochs, $b_1,b_2,b_3$ and $b_1',b_2',b_3'$ (by construction, given the finalization rule). If no blocks share the same epoch number then we can assume w.l.o.g. that $b_3.\mathit{epoch}< b_1'.\mathit{epoch}$. Let block $b'$ be the block with the smallest height that is in the prefix of $b_3'$ such that it has a greater epoch than $b_1$ (such block always exists as $b_1'$ satisfies those conditions). Either $b'.\mathit{height}<b_3.\mathit{height}$ meaning that $b'$ is inconsistent with $b_3$ or else, $b'.\mathit{height}\geq b_3.\mathit{height}$ meaning that the predecessor of $b'$ is inconsistent with $b_1$. Indeed, the predecessor of $b'$ has a strictly smaller height than $b_1$ and by assumption has a smaller epoch number than $b_1$.  Hence there is always a couple of inconsistent blocks in a fork.

Let us now conclude our proof that we solve the eventual strong prefix property. If a fork occurs, then each correct process eventually detects at least one Byzantine process and ignore its votes, hence, we have a finite number of forks as we have a finite number of Byzantine processes, hence eventually there is always a single chain that is finalized. As there is a majority of correct processes, we remain live as in the original Streamlet Protocol. We also inherit its properties for finalizing blocks eventually when synchrony is reached.
\end{proof}
%{\color{red} intuition: in that case, we can use the same algorithm for EP and remove Byzantines once they send two blocks for the same height.}\\
%{\color{red} discussion: let us stress that, we cannot do that with dynamic committee, since we cannot agree on the committee members...to explain more.}

\section{Conclusion}

In this work we closed the gap between eventual prefix and strong prefix consistency. We firstly addressed the question: ``can we  design an asynchronous deterministic blockchain solution to solve bounded displacement eventual prefix ?". We formally showed that  the answer is ``no". On the positive side,  we provided for the first time (i) a solution to eventual strong prefix with a majority of correct processes and (ii) an asynchronous solution to eventual prefix with an unlimited number of Byzantine processes.

\end{document}